\documentclass{llncs}

\usepackage{xspace}

\usepackage{hyphenat}

\usepackage{amsmath}
\usepackage{amssymb}

\newcommand{\nwfs}{{\sc{No\hyp Wait\hyp Flowshop}}\xspace}

\newcommand{\atsp}{{\sc{ATSP}}\xspace}
\newcommand{\atspp}{{\sc{ATSPP}}\xspace}

\newcommand{\e}{\varepsilon}

\newtheorem{observation}[theorem]{Observation}

\newcommand{\eps}{\varepsilon}

\newcommand{\np}{${\cal N\cal P}$}
\newcommand{\nnpp}{${\cal P}={\cal N\cal P}$}

\newcommand{\cmax}{C_{\max}}

\title{No-Wait Flowshop Scheduling is as Hard as Asymmetric Traveling Salesman Problem}
\author{Marcin Mucha\inst{1}\thanks{Part of this work was done while visiting DIMAP at the University of Warwick. Work supported by ERC StG project PAAl no.\ 259515.} \and Maxim Sviridenko\inst{2}\thanks{Work supported by EPSRC grants EP/J021814/1, EP/D063191/1, FP7 Marie Curie Career Integration Grant and Royal Society Wolfson Research Merit Award. }}
\institute{University of Warsaw\\
\email{mucha@mimuw.edu.pl}
\and
University of Warwick\\
\email{M.I.Sviridenko@warwick.ac.uk}
}
\date{}
\begin{document}

\maketitle

\begin{abstract}
In this paper we study the classical no-wait flowshop scheduling problem with makespan objective ($F|no-wait|C_{max}$ in the standard three-field notation). This problem is well-known to be a special case of the asymmetric traveling salesman problem (\atsp) and as such has an approximation algorithm with logarithmic  performance guarantee. In this work we show a reverse connection, we show that any polynomial time $\alpha$-approximation algorithm for the  no-wait flowshop scheduling problem with makespan objective implies the existence of a polynomial-time $\alpha(1+\varepsilon)$-approximation algorithm for the \atsp, for any $\varepsilon>0$. This in turn implies that all non-approximability results for the \atsp (current or future) will carry over to its special case. In particular, it follows that no-wait flowshop problem is APX-hard, which is the first non-approximability result for this problem.
\end{abstract}

\section{Introduction} \label{sec:intro}
\subsection{Problem statement and connection with \atsp}
A {\em flowshop\/} is a multi-stage production process with the
property that all jobs have to pass through several stages. There
are $n$ jobs $J_j$, with $j=1,\ldots,n$, where each job $J_j$ is a
chain of $m$ operations $O_{j1},\ldots,O_{jm}$. Every
operation $O_{ji}$ is preassigned to the  machine $M_i$. The operation $O_{ji}$
has to be processed for $t_{ji}$ time units at its stage; the
value $t_{ji}$ is called its {\em processing time\/} or its {\em
length}.  In a
feasible schedule for the $n$ jobs, at any moment in time every
job is processed by at most one machine and every machine executes
at most one job. For each job $J_j$, operation $O_{ji-1}$ always
is processed before operation $O_{ji}$, and each operation is
processed without interruption on the machine to which it was
assigned.

In the \emph{no-wait flowshop problem} (\nwfs) we require two
additional constraints to be satisfied:
\begin{itemize}
\item There is no waiting time allowed between the execution of consecutive
operations of the same job. Once a job has been started, it has to
be processed without interruption, operation by operation, until
it is completed.
\item Each machine processes the jobs in the same order (i.e.\ we only allow so-called \emph{permutation schedules}).
Note that this applies to all jobs, not only those with non-zero processing time on a given machine.
In other words, one should treat zero length operations as having an infinitely small, but non-zero length.
\end{itemize}

Our goal is to find a permutation $\sigma$ of jobs that
minimizes the {\em makespan\/} (or {\em length\/}) $\cmax(\sigma)$ of the
schedule, i.e., the maximum completion time among all jobs. The
minimum makespan among all feasible schedules is denoted by
$\cmax^*$.

Consider two jobs $J_i$ and $J_j$ that are processed one after another in a no-wait permutation schedule. Let $\delta(i,j)$ be the minimum time we need to wait to start the job $J_j$ after starting the job $J_i$. What is the value of $\delta(i,j)$? Clearly we need to wait at least $t_{i1}$. But since job $j$ cannot wait on the second machine, we also need to wait at least $t_{i1}+t_{i2}-t_{j1}$. Generalizing this leads to the following expression
\begin{equation}
\label{eq:delta-definition}
\delta(i,j) = \max_{q=1,\ldots,m} \left\{ \sum_{k=1}^q t_{ik} - \sum_{k=1}^{q-1} t_{jk}\right\}.
\end{equation}
Note that $\delta$ satisfies the triangle inequality, i.e.\ $\delta(i,j) \le \delta(i,k)+\delta(k,j)$ for any jobs $J_i,J_j,J_k$. The easiest way to see this is by considering the 'waiting time' intuition that led to its definition.

Let $L(j)=\sum_{k=1}^m t_{jk}$ be the total processing time (or length) of   job $J_j$. Then a no-wait schedule that processes the jobs in order $\sigma$ has makespan
\begin{equation}
\label{eq:nwfs-value}
 \cmax(\sigma)=\sum_{k=1}^{n-1} \delta(\sigma_k,\sigma_{k+1}) + L(\sigma_n).
\end{equation}

In the asymmetric traveling salesman problem (\atsp), we are given a complete directed graph $G=(V,E)$ with arc weights $d(u,v)$ for each $u,v\in V$. It is usually assumed that the arc weights satisfy the semimetric properties, i.e. $d(u,u)=0$ for all $u\in V$ and $d(u,v)\le d(u,w)+d(w,v)$ for all $u,w,v\in V$. The goal is to find a Hamiltonian cycle, i.e. a cycle that visits every vertex exactly once, of minimum total weight. The asymmetric traveling salesman path problem (\atspp) is defined analogously, the only difference is that we are looking for a path that starts and ends in arbitrary but distinct vertices and visits all other vertices exactly once along the way.   The distance function $\delta$ can be used to cast \nwfs as \atsp by introducing a dummy job consisting of $m$ zero length operations, and modifying $\delta$ slightly by setting $\delta(i,i)=0$ for all $i=1,\ldots,n$. The role of the dummy job is to emulate the $L(\sigma_n)$ term in~\eqref{eq:nwfs-value}.

The \nwfs was first defined in the 1960 by Piehler \cite{Pi60} who also noticed this problem is a special case of the \atsp. This connection was also later noticed by Wismer \cite{Wismer72}.  The \nwfs is usually denoted $F|no-wait|C_{max}$ using the standard three-field scheduling notation (see e.g.\ Lawler et al.~\cite{LaLeRiSh93}). Although no-wait shop scheduling problems arise naturally in many real-life scenarios (steel manufacturing, hot potato routing) they sometime behave in a way uncommon for other scheduling problems, e.g. speeding up a machine may actually increase the makespan \cite{SW}.

\subsection{Known Results}

For the \nwfs with two machines, the distance matrix of the corresponding \atsp has a very special
combinatorial structure, and the famous subtour patching technique
of Gilmore and  Gomory \cite{GiGo64} yields an $O(n\log n)$ time
algorithm for this case. R{\"o}ck
\cite{Ro84} proves that the three-machine no-wait flowshop is
strongly \np-hard, refining the previous complexity result by
Papadimitriou and Kanellakis \cite{PaKa80} for four machines. Hall
and Sriskandarajah \cite{HaSr96} provide a thorough survey of
complexity and algorithms for various no-wait scheduling models.

We say that a solution to an instance $I$ of a problem is \emph{$\rho$-approximate}
if its value is at most $\rho |OPT|$, where $|OPT|$ is the value of the optimum solution to $I$.
We say that an approximation algorithm has {\em performance
guarantee\/}   $\rho$ for some real
$\rho>1$, if it delivers $\rho$-approximate solutions for all instances. Such an approximation algorithm is then called a
{\em $\rho$-approximation\/} algorithm. A family of polynomial
time $(1+\eps)$-approximation algorithms over all $\eps>0$ is
called a {\em polynomial time approximation scheme\/} (PTAS).

For the   \nwfs with fixed number of machines, i.e.\ $Fm|no-wait|C_{max}$ in standard notation, there exists a polynomial time approximation scheme \cite{S}. The only known approximability results for the general case are $\lceil m/2 \rceil$-approximation algorithm from \cite{RoSc83} or algorithms designed for the \atsp with performance guarantees $\log_2 n$ \cite{FrGaMa82}, $0.999\log_2 n$ \cite{B}, $0.84\log_2 n$ \cite{KLSS}, $0.66\log_2 n$ \cite{FS}, $O\left(\frac{\log n}{\log \log n}\right)$ \cite{AGMGS}.

We remark that the
strongest known negative result for the general \atsp with the
triangle inequality is due to
Karpinski et al.~\cite{Karpinski13}.
They prove that unless \nnpp, the \atsp with
triangle inequality cannot have a polynomial time approximation
algorithm with performance guarantee better than $75/74$.  We are not aware of any known non-approximability results for the \nwfs.

\subsection{Our results and organization of the paper}

In this paper we show that \nwfs is as hard to approximate as \atsp, i.e.\ given an $\alpha$-approximation algorithm for \nwfs one can approximate \atsp with ratio arbitrarily close to $\alpha$. In particular, this gives APX-hardness for \nwfs. It is worth noting that \nwfs has recently received increased interest, since it was viewed as a (potentially) easy case of \atsp, and possibly a reasonable first step towards resolving the general case. It is for this reason that it was mentioned by Shmoys and Williamson~\cite{ShWi} in their discussion of open problems in approximation algorithms. Our results settle this issue.

We also give an $O(\log m)$-approximation algorithm for \nwfs. On one hand, this can be seen as an improvement over the $\lceil m/2 \rceil$-approximation from \cite{RoSc83}. But this result also shows that, unless we obtain an improved approximation for \atsp, the number of machines used by any reduction from \atsp to \nwfs has to be $e^{\Omega(\log n/\log\log n)}$. In this sense our reduction, which uses a number of machines polynomial in $n$, cannot be significantly improved.

The paper is organized as follows. In Section~\ref{sec:hardness} we give the reduction from \atsp to \nwfs. We begin by showing in Subsection~\ref{sec:atsp} that instead of general \atsp instances, it is enough to consider instances of \atspp with integer edge weights that are small relative to $|OPT|$ and polynomial in $n$. We then proceed with the reduction. We start by showing in Subsection~\ref{sec:simple} that any semi-metric can be represented as a \nwfs distance function with only a small additive error. This already shows that \nwfs distance functions are in no way ``easier'' than general semi-metrics. However, this is not enough to reduce \atspp to \nwfs, because of the last term in the objective function~\eqref{eq:nwfs-value}. To make this last term negligible, we blow-up the \atspp instance without significantly increasing the size of the corresponding \nwfs instance, by using a more efficient encoding. This is done in Subsection~\ref{sec:efficient}.

Finally, in Section~\ref{sec:algo} we present the $O(\log m)$-approximation algorithm for \nwfs.

 \section{Non-approximability results for \nwfs}
\label{sec:hardness}

 \subsection{Properties of the \atsp instances}
\label{sec:atsp}
In the rest of the paper we will use $OPT$ to denote an optimal solution of the given \atsp instance and $|OPT|$ the value of such an optimal solution.

\begin{lemma}\label{lem:structure1}
For any instance $G=(V,d)$ of \atsp and any $\varepsilon > 0$, one can construct in time $poly(n,1/\varepsilon)$ another instance $G'=(V',d')$ of \atsp with $|V'| = O(n/\varepsilon)$, such that:
 \begin{enumerate}
 \item  all arc weights in $G'$ are positive integers and the maximal arc weight $W'=O\left(\frac{n\log n}{\varepsilon}\right)$ (regardless of how large the original weights are);
 \item $W'\le \varepsilon |OPT'|$, where $OPT'$ is an optimal solution to $G'$.
 \end{enumerate}
and given an $\alpha$-approximate solution to $G'$ one can construct an $\alpha(1+O(\varepsilon))$-approximate solution to $G$ in time $poly(n,1/\varepsilon)$.
 \end{lemma}
 \begin{proof}
Given an instance $G=(V,d)$ of \atsp, we first run the $\log n$-approximation algorithm for the \atsp from \cite{FrGaMa82}. Let $R$ be the value of the approximate solution found by the algorithm. We know that $|OPT|\le R \le \log n \cdot |OPT|$.

Then we add $\Phi=\frac{\varepsilon R}{n\log_2 n}$ to each arc weight and round each arc weight up to the closest multiple of $\Phi$. Let $\bar{d}(u,v)$ be the new weight of the arc $(u,v)$. We claim that the triangle inequality is still satisfied for new edge weights. Indeed, for any $u,w,v\in V$ we have
\[\bar{d}(u,v)\le d(u,v)+2\Phi\le d(u,w)+d(w,v)+2\Phi\le \bar{d}(u,w)+\bar{d}(w,v).\]
Moreover, the value of any feasible solution  for the two arc weight functions $d$ and $\bar{d}$ differs by at most $2\varepsilon R/\log n \le 2\varepsilon \cdot |OPT|$. We now divide the arc weights in the new instance by $\Phi$. The resulting graph $\hat{G}=(V,\hat{d})$ has integral arc weights. Moreover, they all have values at most $O(\frac{n\log n}{\varepsilon})$, since $d(u,v) \le OPT$ for all $u,v \in V$ by triangle inequality. Finally, any $\alpha$-approximate solution for $\hat{G}$ is also an $\alpha(1+O(\varepsilon))$-approximate solution for $G$.

To guarantee the second property we apply the following transformation to $\hat{G}$. We take $N = \lceil 2/\varepsilon \rceil$ copies of $\hat{G}$. Choose a vertex $u$ in $\hat{G}$ arbitrarily and let $u_1,\dots,u_N$ be the copies of the vertex $u$ in the copies of $\hat{G}$. We define a new graph $G'=(V',d')$ that consists of $N(n-1)+1$ vertices by merging the vertices $u_1,\dots,u_N$ into a supervertex $U$, the remaining vertices of $G'$ consist of $N$ copies of $V\setminus \{u\}$.

If an arc of $G'$ connects two vertices of the same copy of $\hat{G}$ then it has the same weight as the corresponding arc in $\hat{G}$. If an arc $(x_1',x_2')$ connects a copy of a vertex $x_1$ and a copy of a vertex  $x_2$ belonging to different copies of $\hat{G}$ then we define $d'(x_1',x_2')=\hat{d}(x_1,u)+\hat{d}(u,x_2)$, i.e.\ the weight is defined by the travel distance from $w_1$ to $w_2$ through the special supervertex $U$. By definition the maximal weight $W'$ of an arc in $G'$ is at most $2\hat{W}$, where $\hat{W}$ is the maximum weight of an arc in $\hat{G}$, and so the first constraint holds for $G'$.

Moreover, we claim that the value of the optimal Hamiltonian cycle $OPT'$ in $G'$ is exactly $N |\hat{OPT}|$, where $\hat{OPT}$ is an optimum solution for $\hat{G}$. Indeed, it is easy to see that there is a tour of length $\le N |\hat{OPT}|$ obtained by concatenating and short-cutting $N$ optimal tours, one in each copy of $\hat{G}$. On the other hand,  for any feasible tour $T$ in $G'$ we can replace any arc of $T$ that connects vertices (say $w_1$ and $w_2$) in different copies of $\hat{G}$   by two arcs $(w_1,U)$ and $(U,w_2)$. Now we have a walk $\hat{T}$ through $G'$ of the same length as $T$. $\hat{T}$ visits all the vertices of $G'$ exactly once except for the vertex $U$ which is visited multiple times. Therefore, $\hat{T}$ consists of a set of cycles that cover all vertices except $U$ exactly once and vertex $U$ is covered multiple times. We can reorder these cycles so that the walk first visits all vertices of one copy then all vertices of the second copy and so on. By applying short-cutting we obtain a collection of $N$ Hamiltonian cycles, one for each copy of $\hat{G}$. Therefore, the original tour $T$ in $G'$ cannot be shorter than $N |\hat{OPT}|$, and so $|OPT'| = N |\hat{OPT}|$.

We now have
\[ W' \le 2\hat{W} \le 2|\hat{OPT}| = 2|OPT'|/N \le \varepsilon |OPT'|,\]
so the second constraint is satisfied. The above argument is constructive, i.e.\ given a Hamiltonian cycle of length $L$ in $G'$, it produces a Hamiltonian cycle in $G$ of length at most $L/N$ in time $poly(n,1/\varepsilon)$.\qed
 \end{proof}

\begin{lemma}
\label{lem:atsp-to-atspp}
Let $G=(V,d)$ be an instance of \atsp with $|V|=n$ and $d:V\times V \rightarrow \{0,\ldots,W\}$. Then, one can construct in time $O(n)$ an instance $G'=(V',d')$ of \atspp with $|V'|=n+1$ and $d':V\times V \rightarrow \{0,\ldots,2W\}$, such that the optimal values of the two instances are the same. Moreover, given a solution $S'$ of $G'$, one can construct in time $O(n)$ a solution of $G$ with value at most the value of $S'$.
\end{lemma}

\begin{proof}
We fix a vertex $v \in V$ and define $G'$ as follows:
\begin{itemize}
\item $V' = V \setminus \{v\} \cup \{v_{in},v_{out}\}$, i.e.\ we split $v$ into two vertices.
\item For all pairs $x,y \in V \setminus \{v\}$ we put $d'(x,y) = d(x,y)$.
\item For all $x \in V \setminus \{v\}$ we put $d'(v_{out},x) = d(v,x)$ and $d'(x,v_{in}) = d(x,v)$, i.e.\ $v_{in}$ inherits the incoming arcs of $v$ and $v_{out}$ inherits the outgoing arcs. We also put $d'(v_{out},v_{in})=0$.
\item All the remaining arcs get length of $2W$.
\end{itemize}
It is easy to verify that $d'$ satisfies the triangle inequality.

We now need to show that the shortest Hamiltonian tour in $G$ has the same length as the shortest Hamiltonian path in $G'$. Note that Hamiltonian tours in $G$ correspond to Hamiltonian paths in $G'$ starting in $v_{out}$ and ending in $v_{in}$, and that this correspondence maintains the total length. Using this observation, for any tour in $G$, one can obtain a path in $G'$ of the same length.

In the opposite direction, let us consider a path $S'$ in $G'$. We will show how to transform $S'$ without increasing its length, so that it begins in $v_{out}$ and ends in $v_{in}$. We proceed in two steps. First, if $S'$ does not begin in $v_{out}$, we break it before $v_{out}$ and swap the order of the two resulting subpaths. This does not increase the length since any incoming arc of $v_{out}$ has length $2W$.

Now, suppose that $v_{in}$ is not the last vertex on $S'$, i.e.\ it is visited between $x$ and $y$ for some $x,y \in V$ ($v_{in}$ cannot be the first vertex, since $v_{out}$ is). We remove $v_{in}$ from $S'$ and append it on the end. The total change in the length of the path is
\[ \Delta = d'(x,y) + d'(z,v_{in}) - d'(x,v_{in}) - d'(v_{in},y),\]
where $z$ is the last vertex of the path. Using the definition of $d'$ we get
\[ \Delta \le W + W - 0 - 2W \le 0,\]
so this transformation does not increase the length of the path, which ends the proof.\qed
\end{proof}

\subsection{ A Simple Embedding}
\label{sec:simple}

In this section we show that jobs with the distance function $\delta$ in some sense form a universal space for all semi-metrics (approximately).
More precisely, let $\mathcal{J}_{m,T}$ be the set of all $m$-machine jobs with all operations of length at most $T$, i.e.\ $\mathcal{J}_{m,T} = \{0,1,\ldots,T\}^m$. Then
\begin{theorem}
\label{thm:main}
For any $n$-point semi-metric $(V,d)$, where $d:V\rightarrow \{0,..,D\}$, there exists a mapping $f:V\rightarrow \mathcal{J}_{2nD,1}$, such that
\[ \delta(f(u),f(v)) = d(u,v)+1 \textrm{ for all } u,v \in V, \]
where $\delta$ is the distance function defined by~\eqref{eq:delta-definition}.
\end{theorem}
\begin{proof}
 First we define a collection of $D+1$ jobs ${\cal J}(D)=\{B_0^D,...,B_D^D\}$ on $2D$ machines with all operations of length either zero or one. Obviously, ${\cal J}(D)\subseteq J_{2D,1}$. The job $B_i^D$ consists of $D-i$ zero length operations that must be processed on machines $M_1,\dots, M_{D-i}$, followed by $D$ unit length operations that must be processed on machines $M_{D-i+1},\dots,M_{2D-i}$. The last $i$ operations have zero length. By construction,  $L(B_i^D)=D$ for  $i=0, \dots, D$. Moreover, $\delta(B_i^D,B_j^D) = \max(i-j+1,0)$.

In the following, we will use the symbol $\cdot$ to denote concatenation of sequences, and in particular sequences of jobs.
 Let $a_i\in {\cal J}(D)$ and $b_i\in {\cal J}(D)$ for $i=1,\dots,k$. Consider the job $A=a_1\cdot a_2\cdot \ldots \cdot a_k\in J_{k2D,1}$ processed on $k2D$ machines $M'_1,\dots,M'_{k2D}$. That is, job $A$ has the same operation length on machine $M'_{(i-1)2D+r}$ as job $a_i$ on machine $M_r$ for $r=1,\dots, 2D$. Analogously, let $B=b_1\cdot b_2\cdot \ldots \cdot b_k \in J_{k2D,1}$.
 Then
\begin{equation}\label{Lem2}
 \delta(A, B) = \max\{\delta(a_1,b_1),\ldots,\delta(a_k,b_k)\}.
 \end{equation}
This is because by making job $B$ start $X$ time steps after job $A$, we ensure that each sequence $b_i$ of operations starts $X$ time steps after the sequence $a_i$ of operations. Therefore, in any feasible schedule $X\le \max\{\delta(a_1,b_1),\ldots,\delta(a_k,b_k)\}$. On the other side, starting job $B$ exactly $\max\{\delta(a_1,b_1),\ldots,\delta(a_k,b_k)\}$ time steps after the start of the job $A$ gives a feasible schedule.

Let $V=\{v_1,...,v_n\}$. We define $f(v_i) = B_{d(v_i,v_1)}^D \cdot\ldots \cdot B_{d(v_i,v_n)}^D$. Then by  (\ref{Lem2})
we have
\[ \delta(f(v_i),f(v_j)) = \max_k \left\{ d(v_i,v_k)-d(v_j,v_k)+1,0\right\}=d(v_i,v_j)+1.\]
The last equality follows from the triangle inequality and the fact that $d(v_j,v_j)=0$.\qed
\end{proof}

\subsection{A More Efficient Embedding}
\label{sec:efficient}
Our main result concerning the relationship between \atsp and \nwfs is the following.
\begin{theorem}\label{thm:atsp-to-nwfs}
Let $G=(V,d)$ be an instance of \atsp with $|V|=n$  and let $OPT$ be the optimum TSP tour for $G$. Then, for any constant $\e > 0$, there exists an instance $I$ of \nwfs, such that given an $\alpha$-approximate solution to $I$, we can find a solution to $G$ with length at most
\[ \alpha(1+O(\e))|OPT|.\]
Both $I$ and the solution to $G$ can be constructed in time $poly(n,1/\e)$.
\end{theorem}
\begin{proof}
 We start by applying Lemma~\ref{lem:structure1} to $G$ and then Lemma~\ref{lem:atsp-to-atspp} to the resulting instance of \atsp.
Finally, we scale all the distances up by a factor of $\lceil 1/\varepsilon\rceil$.
In this way we obtain an instance $G'=(V',d')$ of \atspp, such that:
\begin{itemize}
\item $n' = |V'| = O(n/\varepsilon)$.
\item $n' \le \varepsilon OPT'$ (this is due to scaling, since all arc weights are positive integers before scaling).
\item $G'$ has integral arc weights.
\item $W' = O(n\log_2 n/\varepsilon^2)$ and $W' \le \varepsilon |OPT'|$, where $W'$ and $OPT'$ are the maximum arc weight and the optimum solution for $G'$, respectively.
\item Given an $\alpha$-approximate solution to $G'$, one can obtain an $\alpha(1+O(\varepsilon))$-approximate solution to $G$ (using Lemma~\ref{lem:structure1} and Lemma~\ref{lem:atsp-to-atspp}).
\end{itemize}

Note that one can simply encode $G'$ as a \nwfs instance using Theorem~\ref{thm:main}. The problem with this approach is that the objective value in \nwfs contains an additional term that is not directly related to the   distances in $G'$. If this term dominates the makespan, approximation algorithms for \nwfs are useless for the original \atspp instance.

To overcome this obstacle we first blow $G'$ up by creating $N$ (to be chosen later) copies of it. Let $G_1=(V_i,d_i),\ldots,G_N=(V_N,d_N)$ be these copies. We join these copies into a single instance $\hat{G}=(\hat{V},\hat{d})$ by putting edges of length $2W'$ between all pairs of vertices from different copies. Note that any TSP path in $\hat{G}$ can be transformed into a path in which vertices of the same copy form a subpath, without increasing its cost. Therefore
\begin{observation}
The cost of the optimum solution for $\hat{G}$ is $N|OPT'|+2W'(N-1)$. In the opposite direction, given a TSP path of cost $C$ in $\hat{G}$, one can obtain a TSP path of cost $\frac{C-2W'(N-1)}{N}$ in $G'$.
\end{observation}
We transform $\hat{G}$ into a \nwfs instance as follows. We apply the construction of Theorem~\ref{thm:main} to each $G_i$ to obtain $N$ identical jobsets $J_1,\ldots,J_N$. We then augment these jobs to enforce correct distances between jobs in different $J_i$. To this end we introduce new gadgets.
\begin{lemma}
\label{lem:gadget2}
For any $N \in \mathbb{N}$ there exists a set of $N$ jobs $H_1,\ldots,H_N$, each of the jobs using the same number of machines $O(D\log N)$ and of the same total length $O(D \log N)$, such that $\delta(H_i,H_i)=1$ and $\delta(H_i,H_j) = D$ for $i \neq j$.
\end{lemma}
\begin{proof}
We will use the following two jobs as building blocks: $H^0 = (10)^{2D}$ and $H^1 = 1^{2D}0^{2D}$ ($x^D$ here means a sequence constructed by repeating the symbol $x$ exactly $D$ times). Note that they have the same total length of $2D$, the same number of machines $4D$, and that $\delta(H^0,H^1) = \delta(H^0,H^0)=\delta(H^1,H^1)=1$ and $\delta(H^1,H^0)=D$.

Let $k$ be smallest integer such that ${2k \choose k} \ge N$. Clearly $k=O(\log N)$. Consider characteristic vectors of all $k$-element subsets of $\{1,\ldots,2k\}$, pick $N$ such vectors $R_1,\ldots,R_N$. Now, construct $H_i$ by substituting $H^0$ for each $0$ in $R_i$ and $H^1$ for each $1$. Analogously to (\ref{Lem2}), we derive that the distances between $H_i$ are as claimed. Also, the claimed bounds on the sizes of $H_i$ follow directly from the construction.\qed
\end{proof}

Using the above lemma it is easy to ensure correct distances for jobs in different $J_i$. Simply augment all jobs with gadgets described in Lemma~\ref{lem:gadget2}, same gadgets for the same $J_i$, different gadgets for different $J_i$. Here $D=2W'+1$, so the augmentation only requires $O(W'\log N)$ extra machines and extra processing time.

This ends the construction of the instance of \nwfs. The optimum solution in this instance has cost
\[ N|OPT'|+2W'(N-1)+(Nn'-1) + 2W'n'+O(W'\log N).\]
The $Nn'-1$ term here comes from the additive error in Theorem~\ref{thm:main}, and the $2W'n'+O(W'\log N)$ term corresponds to the processing time of the last job in the optimum solution.

Given an $\alpha$-approximate solution to the flowshop instance, we can obtain a TSP path $ALG_{\hat{G}}$ for $\hat{G}$ with cost
\begin{align*}
 |ALG_{\hat{G}}| \le \alpha (N|OPT'|+2W'(N-1)+(Nn'-1) + 2W'n'+O(W\log N)) \\
- (Nn'-1) - 2W'n'-O(W'\log N),
\end{align*}
which is just
\[ \alpha (N|OPT'|+2W'(N-1)) + (\alpha-1)( (Nn'-1) + 2W'n' + O(W'\log N) ).\]
As observed earlier, from this we can obtain a solution to $G'$ with value at most
\[\frac{|ALG_{\hat{G}}|-2W'(N-1)}{N}\]
which is bounded by
\[\frac{\alpha N|OPT'|+(\alpha-1)(2W'(N-1)+(Nn'-1) + 2W'n' +O(W'\log N)) }{N} .\]
By taking $N=n'$ we can upper-bound this expression by
\[\alpha|OPT'| + (\alpha-1)(2W'+n') + 2W' + O\left(\frac{W'\log N}{N}\right). \]
Using the fact that $\max\{W',n'\}\le \varepsilon |OPT'|$ we can upper-bound this by
\[\alpha|OPT'| + O(\alpha\e)|OPT'| = \alpha(1+O(\varepsilon))|OPT'|.\]
As noted earlier, we can transform this $\alpha(1+O(\varepsilon))$-approximate solution to $G'$ into a $\alpha(1+O(\varepsilon))$-aproximate solution to $G$.

As for the running time, it is polynomial in the size of the \nwfs instance constructed. We have $O(Nn') = O(n^2/\varepsilon^2)$ jobs in this instance, and $O(W' n') + O(W'\log N) = O(n^2\log n / \varepsilon^3)$ machines, so the running time
is $poly(n,1/\varepsilon)$.\qed
\end{proof}

Using the result of the Karpinski et al.~\cite{Karpinski13} for the  \atsp we derive
\begin{theorem}
\nwfs is not approximable with factor better than $\frac{75}{74}$, unless $P=NP$.
\end{theorem}

 \section{An $O(\log m)$-Approximation Algorithm for \nwfs.}
\label{sec:algo}

\begin{theorem}
There exists an $O(\log m)$-approximation algorithm for \nwfs.
\end{theorem}

\begin{proof}
Consider any instance $I$ of \nwfs. Let $G=G_I$ be the \atsp instance resulting from a standard reduction from \nwfs to \atsp, i.e.\ $G$ is obtained by adding a dummy all-zero job to $I$ and using $\delta$ as the distance function.

Our algorithm is a refinement of the approximation algorithm of Frieze, Galbiati and Maffioli~\cite{FrGaMa82}.
This algorithm starts by finding a minimum cost cycle cover $C_0$ in $G$.
 Since $OPT$ is a cycle cover we know that $|C_0|\le |OPT|$.
 After that we choose a single vertex from each cycle -- one that corresponds to the shortest job (i.e. we choose a vertex $j$ with smallest $L(j)$). We then take $G_1$ to be the subgraph of $G$ induced by the selected vertices. As was noted in \cite{FrGaMa82} the optimal \atsp solution ($OPT_1$) for $G_1$ is at most as long as $OPT$, i.e.\ $|OPT_1|\le |OPT|$.

We now reiterate the above procedure: We find a minimum cycle cover $C_1$ in $G_1$. We again have $|C_1|\le |OPT_1|$. We choose a single vertex per cycle of $C_1$, again corresponding to the job with smallest length, define $G_2$ to be the subgraph of $G_1$ induced by the selected vertices, and so on. In each iteration we decrease the cardinality of the set of vertices by a factor of at least two. If the cycle cover $C_i$ consists of a single cycle for some $i=0,\dots, \log_2 m -1$, we consider a subgraph of $G$ that is the union of all the $C_k$ for $k=0,\dots,i$. This is an Eulerian subgraph of $G$ of cost at most $\log_2 m\cdot |OPT|$, and  can be transformed into a feasible Hamiltonian cycle by the standard procedure of short-cutting.

  Otherwise, the graph $G_{\log_2 m}$ consists of more than one vertex. Let $C=\cup_{i=1}^{\log_2m -1}C_i$ be the subgraph of $G$ that is the union of $\log_2 m$ cycle covers $C_k$ for $k=0,\dots, \log_2 m -1$. Consider the subgraph $G'$ of $G$ that is the union of $C$ and an arbitrary Hamiltonian cycle $H'$ in $G_{\log_2 m}$. Each connected component of $C$ consists of at least $m$ vertices. The vertex set of $G_{\log_2 m}$ consists of vertices corresponding to the shortest jobs in each of the connected components of $C$. Let $S$ be this set of jobs. We now claim that the length of $H'$ is at most
  \[\sum_{j\in S} L(j)\le \frac{1}{m}\sum_{j=1}^nL(j)\le \frac{1}{m}m\cdot C^*_{max}=C^*_{max}.\]
The last inequality follows from the fact that sum of processing times of all operations that must be processed on a single machine is a lower bound on the value of the optimal makespan. It follows that the total weight of $G'$ is at most $\log_2m+1$ times the optimal makespan. By shortcutting we can construct a Hamiltonian cycle in $G$, which in turn gives us an approximate solution for the original instance $I$ of \nwfs.\qed
\end{proof}

\section{Acknowledgments}
The first author would like to thank DIMAP, and in particular Artur Czumaj, for making his visit to the University of Warwick possible.
\bibliographystyle{splncs03}
\bibliography{no-wait}

\end{document}